\documentclass[10pt,journal,epsfig]{IEEEtran}

\usepackage[dvips]{graphicx}
\usepackage{graphicx}
\usepackage{amssymb}
\usepackage{cite}
\usepackage{subfigure}
\usepackage{amsmath}

\usepackage{algorithm}
\usepackage{algorithmic}

\begin{document}

\title{Millimeter Wave Channel Estimation via Exploiting Joint Sparse and Low-Rank Structures}

\author{Xingjian Li, Jun Fang, Hongbin Li,~\IEEEmembership{Senior
Member,~IEEE}, and Pu Wang
\thanks{Xingjian Li and Jun Fang are with the National Key Laboratory
of Science and Technology on Communications, University of
Electronic Science and Technology of China, Chengdu 611731, China,
Email: JunFang@uestc.edu.cn}
\thanks{Hongbin Li is
with the Department of Electrical and Computer Engineering,
Stevens Institute of Technology, Hoboken, NJ 07030, USA, E-mail:
Hongbin.Li@stevens.edu}
\thanks{Pu Wang is
with the Mitsubushi Electric Research Laboratories, Cambridge, MA
02139, USA, E-mail: pwang@merl.com}
\thanks{This work was supported in part by the National Science
Foundation of China under Grant 61522104, and the National Science
Foundation under Grant ECCS-1408182 and Grant ECCS-1609393.}}

\maketitle

%Channel estimation in mmWave systems is challenging as the analog
%beamforming/combining structure prevents the digital baseband from
%directly accessing the entire channel dimension. This is also
%referred to as the channel subspace sampling limitation, which
%makes it difficult to acquire useful channel state information
%(CSI) during a practical channel coherence time.

%and formulate the mmWave channel estimation problem as a sparse
%signal recovery problem

\begin{abstract}
We consider the problem of channel estimation for millimeter wave
(mmWave) systems, where, to minimize the hardware complexity and
power consumption, an analog transmit beamforming and receive
combining structure with only one radio frequency (RF) chain at
the base station (BS) and mobile station (MS) is employed. Most
existing works for mmWave channel estimation exploit sparse
scattering characteristics of the channel. In addition to
sparsity, mmWave channels may exhibit angular spreads over the
angle of arrival (AoA), angle of departure (AoD), and elevation
domains. In this paper, we show that angular spreads give rise to
a useful low-rank structure that, along with the sparsity, can be
simultaneously utilized to reduce the sample complexity, i.e. the
number of samples needed to successfully recover the mmWave
channel. Specifically, to effectively leverage the joint sparse
and low-rank structure, we develop a two-stage compressed sensing
method for mmWave channel estimation, where the sparse and
low-rank properties are respectively utilized in two consecutive
stages, namely, a matrix completion stage and a sparse recovery
stage. Our theoretical analysis reveals that the proposed
two-stage scheme can achieve a lower sample complexity than a
direct compressed sensing method that exploits only the sparse
structure of the mmWave channel. Simulation results are provided
to corroborate our theoretical results and to show the superiority
of the proposed two-stage method.
\end{abstract}

%over the conventional compressed sensing method


%can be put into practice

\section{Introduction}
Millimeter wave (mmWave) communication is a promising technology
for future 5G cellular networks
\cite{RappaportMurdock11,RanganRappaport14,GhoshThomas14}. It has
the potential to offer gigabits-per-second communication data
rates by exploiting the large bandwidth available at mmWave
frequencies. However, a key challenge for mmWave communication is
that signals incur a much more significant path loss over the
mmWave frequency bands as compared with the path attenuation over
the lower frequency bands \cite{SwindlehurstAyanoglu14}. To
compensate for the significant path loss, large antenna arrays
should be used at both the base station (BS) and the mobile
station (MS) to provide sufficient beamforming gain for mmWave
communications \cite{AlkhateebMo14}.

Although directional beamforming helps overcome the path loss
issue, it also complicates the mmWave communication system design.
Due to the narrow beam of the antenna array, communication between
the transmitter and the receiver is possible only when the
transmitter's and receiver's beams are well-aligned, i.e. the beam
directions are pointing towards each other. Therefore beamforming
training is required to search for the best beamformer-combiner
pair that gives the highest channel gain. One method is to
exhaustively search for all possible beam pairs to identify the
best beam alignment. Nevertheless, this exhaustive search may lead
to a prohibitively long training process, particularly when the
number of antennas at the BS and MS is large. To address this
issue, an adaptive beam alignment algorithm was proposed in
\cite{HurKim13}, where a hierarchical multi-resolution beamforming
codebook was employed to avoid the costly exhaustive sampling of
all pairs of transmit and receiver beams. Nevertheless, this
adaptive beam alignment requires a feedback channel from the
receiver to the transmitter, which may not be available before the
communication between the receiver and the transmitter is
established. Recently, a novel beam steering scheme called as
``Agile-Link'' \cite{AbariHassanieh16} was proposed to find the
correct beam alignment without scanning the space. The main idea
of the Agile-Link is to harsh the beam directions using a few
carefully chosen hash functions, and steer the antenna array to
beam along multiple directions simultaneously.

%Due to the narrow beam of the antenna array, beam alignment
%between the transmitter and the receiver is needed to establish
%communication. Specifically, the objective of beam alignment is to
%search for the best beamformer-combiner pair by letting the
%transmitter and receiver scan the

%The beam searching approaches in can be viewed as channel path
%searching techniques.

Unlike beam scanning techniques whose objective is to find the
best beam pair, another approach is to directly estimate the
mmWave channel or its associated parameters, e.g. angles of
arrival/departure, e.g.
\cite{RamasamyVenkateswaran12a,RamasamyVenkateswaran12b,
AlkhateebyLeus15,AlkhateebAyach14,SchniterSayeed14,KimLove15,
MarziRamasamy16,GaoDai15b,GaoDai16,ZhouFang16}. In particular, by
exploiting the sparse scattering nature of mmWave channels, mmWave
channel estimation can be formulated as a sparse signal recovery
problem
\cite{AlkhateebyLeus15,AlkhateebAyach14,SchniterSayeed14,KimLove15,
MarziRamasamy16,GaoDai15b,GaoDai16}, and it has been shown that
substantial reduction in training overhead can be achieved.
Besides the compressed sensing techniques, low-rank tensor
factorization methods \cite{ZhouFang16,ZhouFang17} were recently
proposed to exploit the low-rank structure of mmWave channels, and
have been shown to outperform the compressed sensing-based methods
in terms of both estimation accuracy and computational complexity.

%The tensor factorization-based method, however, may require some
%setup at the transmitter such that the received signal can be
%organized as a tensor that admits a CANDECOMP/PARAFAC (CP)
%decomposition.

%Note that most existing mmWave channel estimation methods (e.g.
%\cite{AlkhateebyLeus15,AlkhateebAyach14,SchniterSayeed14,KimLove15,ZhouFang16})
%only exploit the sparse characteristic of mmWave channels.

%As a result, the mmWave channel may exhibit a simultaneously
%sparse and low-rank structure that can be utilized to achieve a
%higher training overhead reduction.

%in order to better exploit the simultaneously sparse and low-rank
%structure inherent in mmWave channels

In addition to the sparse scattering characteristic, several
real-world measurements in dense-urban propagation environments
(e.g.
\cite{SamimiWang13,ZhaoMayzus13,RappaportSun13,AkdenizLiu14})
reveal that mmWave channels spread in the form of clusters of
paths over the angular domains including the angle of arrival
(AoA), angle of departure (AoD), and elevation. In
\cite{RappaportSun13,AkdenizLiu14}, real-world channel
measurements at 28 and 73 GHz in New York city were reported, in
which the angular spread has been explicitly studied in terms of
the root mean-squared (rms) beamspread in the different angular
(AoA, AoD, and elevation) dimensions. Specifically, the measured
AoA spreads (in terms of rms) are $15.5^\circ$ and $15.4^\circ$,
respectively, for the two carrier frequencies, while the measured
AoD spreads (in terms of rms) are $10.2^\circ$ and $10.5^\circ$,
respectively. Moreover, the angular spread increases as the
spatial resolution becomes finer when the number of antennas at
the BS/MS increases. As demonstrated in \cite{WangPajovic17}, the
angular spreads give rise to a block-sparse structure that can be
exploited to improve the mmWave channel estimation performance.

In this paper, we further show that, in the presence of angular
spreads, the mmWave channel exhibits a joint sparse and low-rank
structure. To better utilize the joint structure, we propose a
two-stage compressed sensing scheme, where a low-rank matrix
completion stage is first performed and then followed by a
compressed sensing stage to recover the mmWave channel. Our
analysis reveals that the number of measurements required for
exact channel recovery is $\mathcal{O}(p L^2)$ for the proposed
two-stage method, where $L$ represents the number of scattering
clusters and $p$ is a quantity that measures the maximum angular
spread among all scattering clusters. While a direct compressed
sensing method that exploits only the sparsity of mmWave channels
requires a number of measurements of $\mathcal{O}(p^2 L)$. Thus
the proposed two-stage compressed sensing method achieves a lower
sample complexity than the direct compressed sensing method when
$L<p$, which is very likely to hold in dense-urban propagation
environments where the angular spreads over the AoA/AoD/elevation
domains could be relatively large.

%and this joint structure can be utilized to improve the sample
%complexity

%reveals that for mmWave channels with moderate or large angular
%spreads, the proposed two-stage scheme requires fewer measurements
%than the direct compressed sensing scheme for exact channel
%reconstruction. This result is also corroborated by our simulation
%results, where it is shown the proposed two-stage scheme achieves
%a performance improvement over the direct compressed sensing
%scheme, especially when the angular spread becomes large.

The rest of the paper is organized as follows. The system model
and the problem formulation are discussed in Section
\ref{sec:system-model}. In Section \ref{sec:channel-model}, we
introduce a geometric mmWave channel model with angular spreads
and show that the mmWave channel exhibits a joint sparse and
low-rank structure. A two-stage compressed sensing method is
developed in Section \ref{sec:proposed-method}, along with a
theoretical analysis for the two-stage method. Simulation results
are provided in Section \ref{sec:experiments}, followed by
concluding remarks in Section \ref{sec:conclusion}.

%The rest of the paper is organized as follows. In Section
%\ref{sec_system}, we introduces the system model, the sparse
%representation of the mmWave channel and the formulation of direct
%compressed sensing method. In Section \ref{sec_two-stage}, the
%procedure of our proposed algorithm is illustrated and the
%required number of measurements is analyzed and proved. Simulation
%results are provided in Section \ref{sec_simulation}, followed by
%conclusions in Section \ref{sec_conclusion}.

\begin{figure}[!t]
\centering
\includegraphics[width=3.5in]{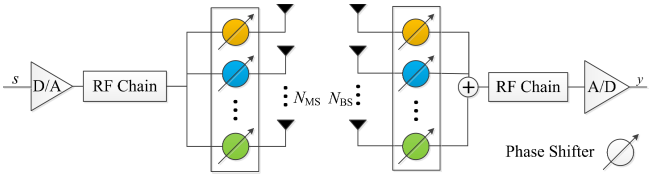}
\caption{A block diagram of the analog transmit beamforming and
receive combining structure.} \label{fig1}
\end{figure}

\section{System Model and Prior Work} \label{sec:system-model}
Consider a point-to-point mmWave MIMO system consisting of
$N_{\text{BS}}$ antennas at the BS and $N_{\text{MS}}$ antennas at
the MS. Since the radio frequency (RF) chains are costly and
power-consuming at mmWave frequency bands, to minimize the
hardware complexity and power consumption, we focus on an analog
transmit beamforming and receive combining structure (see Fig.
\ref{fig1}) where only one RF chain is employed at the BS and MS.
In this structure, transmit beamforming and receive combining are
implemented in the analog domain using digitally controlled phase
shifters. At time instant $t$, the transmitter employs a
beamforming vector
$\boldsymbol{f}(t)\in\mathbb{C}^{N_{\text{MS}}}$ to transmit a
symbol $s(t)$, and at the receiver, the received signals on all
antennas are combined with a receive combining vector
$\boldsymbol{z}(t)\in\mathbb{C}^{N_{\text{BS}}}$. The combined
signal at the receiver can therefore be expressed as
\begin{align}
y(t)=\boldsymbol{z}^{H}(t)\boldsymbol{H}\boldsymbol{f}(t)s(t)+w(t)
\quad \forall t=1,\ldots,T \label{data-model}
\end{align}
where $\boldsymbol{H}\in\mathbb{C}^{N_{\text{BS}}\times
N_{\text{MS}}}$ is the channel matrix, and $w(t)$ denotes the
additive Gaussian noise with zero mean and variance $\sigma^2$.
Without loss of generality, we set $s(t)=1$ during the training
phase. Note that since the precoder and combiner are implemented
by analog phase shifters, entries of $\boldsymbol{z}(t)$ and
$\boldsymbol{f}(t)$ have constant modulus.

%of the channel rather than the $\boldsymbol{H}$ directly

%In , the channel estimation problem is converted to a general
%beam-alignment problem, where the best beamforming codeword is
%determined according to the value of beamforming gain.
%Specifically, the optimal beam pair
%$\{\boldsymbol{z},\boldsymbol{f}\}$ is selected from the
%predesigned codebooks $\mathcal{Z}$ and $\mathcal{F}$ by
%maximizing the value of
%$\|\boldsymbol{z}^{H}\boldsymbol{H}\boldsymbol{f}\|^2$.
%Unfortunately, if we exhaustively search over all possible
%beamformer and combiner combinations, the training overhead will
%be unacceptable since the size of codebook $\mathcal{Z}$ and
%$\mathcal{F}$ is very large in massive MIMO systems. To save the
%time resource during training phase, a hierarchical subcodebook
%with much smaller size is designed in \cite{Hur2013Millimeter},
%within each subcodebook an adaptive sampling strategy is used to
%select the optimal sounding beam pair.

We see that in mmWave systems, the receiver cannot directly
observe $\boldsymbol{H}$, rather it observes a noisy version of
$\boldsymbol{z}^{H}\boldsymbol{H}\boldsymbol{f}$. This is also
referred to as the channel subspace sampling limitation
\cite{HurKim13,KimLove15}, which makes channel estimation a
challenging problem. By exploiting the sparse scattering nature of
mmWave channels, the channel estimation problem can be formulated
as a sparse signal recovery problem (e.g.
\cite{HurKim13,KimLove15}). Specifically, note that the mmWave
channel is usually characterized by a geometric channel model
(see, e.g. \cite{AlkhateebAyach14})
\begin{align}
\boldsymbol{H} =
\sum_{l=1}^L\alpha_{l}\boldsymbol{a}_{\text{BS}}(\theta_{l})\boldsymbol{a}_{\text{MS}}^H(\phi_{l})
\label{channel-model}
\end{align}
where $L$ is the number of paths, $\alpha_{l}$ is the complex gain
associated with the $l$th path, $\theta_{l}\in[0,2\pi]$ and
$\phi_{l}\in[0,2\pi]$ are the associated azimuth AoA and azimuth
AoD respectively, and
$\boldsymbol{a}_{\text{BS}}\in\mathbb{C}^{N_{\text{BS}}}$
($\boldsymbol{a}_{\text{MS}}\in\mathbb{C}^{N_{\text{MS}}}$) is the
array response vector associated with the BS (MS). Suppose a
uniform linear array (ULA) antenna array is used. Then the
steering vectors at the BS and the MS can be written as
\begin{align}
&\boldsymbol{a}_{\text{BS}}(\theta_{l})=\frac{1}{\sqrt{N_{\text{BS}}}}
\left[1,e^{j\frac{2\pi}{\lambda}d\sin(\theta_{l})},\ldots,e^{j(N_{\text{BS}}-1)
\frac{2\pi}{\lambda}d\sin(\theta_{l})}\right]^T \nonumber\\
&\boldsymbol{a}_{\text{MS}}(\phi_{l})=\frac{1}{\sqrt{N_{\text{MS}}}}
\left[1,e^{j\frac{2\pi}{\lambda}d\sin(\phi_{l})},\ldots,e^{j(N_{\text{MS}}-1)\frac{2\pi}{\lambda}d\sin(\phi_{l})}\right]^T
\nonumber
\end{align}
where $\lambda$ is the signal wavelength, and $d$ is the distance
between neighboring antenna elements. To formulate the channel
estimation as a sparse signal recovery problem, we first express
the channel as a beam space MIMO representation as follows
\begin{align}
\boldsymbol{H} =
\boldsymbol{A}_{\text{BS}}\boldsymbol{H}_v\boldsymbol{A}_{\text{MS}}^H
\label{virtual-channel}
\end{align}
where $\boldsymbol{A}_{\text{BS}}\triangleq
[\boldsymbol{a}_{\text{BS}}(\psi_1),\ldots,\boldsymbol{a}_{\text{BS}}(\psi_{N_1})]$
is an overcomplete matrix ($N_1\geq N_{\text{BS}}$) with each
column a steering vector parameterized by a pre-discretized AoA,
$\boldsymbol{A}_{\text{MS}}\triangleq
[\boldsymbol{a}_{\text{MS}}(\omega_1),\ldots,\boldsymbol{a}_{\text{MS}}(\omega_{N_2})]$
is an overcomplete matrix (i.e. $N_2\geq N_{\text{MS}}$) with each
column a steering vector parameterized by a pre-discretized AoD,
and $\boldsymbol{H}_v\in \mathbb{C}^{N_1\times N_2}$ is a sparse
matrix with $L$ non-zero entries corresponding to the channel path
gains $\{\alpha_{l}\}$. Here the true AoA and AoD parameters are
assumed to lie on the discretized grids.

Substituting (\ref{virtual-channel}) into (\ref{data-model}), we
have
\begin{align}
y(t) &=
\boldsymbol{z}^{H}(t)\boldsymbol{A}_{\text{BS}}\boldsymbol{H}_v
\boldsymbol{A}_{\text{MS}}^H\boldsymbol{f}(t) +w(t) \nonumber \\
&= \left[(\boldsymbol{A}_{\text{MS}}^H\boldsymbol{f}(t))^T\otimes
(\boldsymbol{z}(t)^{H}\boldsymbol{A}_{\text{BS}})\right]\boldsymbol{h} +w(t)\nonumber \\
&= (\boldsymbol{f}(t)^T\otimes
\boldsymbol{z}(t)^{H})(\boldsymbol{A}_{\text{MS}}^*\otimes
\boldsymbol{A}_{\text{BS}})\boldsymbol{h}+w(t)
\end{align}
where $\otimes$ denotes the Kronecker product, $()^{\ast}$
represents the complex conjugate, and $\boldsymbol{h}\triangleq
\text{vec}(\boldsymbol{H}_v)$. Collecting all measurements
$\{y(t)\}$ and stacking them into a vector
$\boldsymbol{y}\triangleq [y_1\phantom{0}\ldots\phantom{0}y_T]^T$,
we arrive at
\begin{align}
\boldsymbol{y} &= \left[
\begin{array}{c}
                       (\boldsymbol{f}(1)^T\otimes \boldsymbol{z}(1)^{H}) \\
                       \vdots \\
                       (\boldsymbol{f}(T)^T\otimes \boldsymbol{z}(T)^{H})
\end{array}
\right](\boldsymbol{A}_{\text{MS}}^*\otimes \boldsymbol{A}_{\text{BS}})\boldsymbol{h}+\boldsymbol{w} \nonumber \\
&\triangleq \boldsymbol{\Psi}\boldsymbol{h} +\boldsymbol{w}
\end{align}
Estimating $\boldsymbol{h}$ now can be formulated as a sparse
signal recovery problem as follows
\begin{align}
\text{min} & \quad\|\boldsymbol{h}\|_1 \nonumber\\
\text{s.t.} & \quad
\|\boldsymbol{y}-\boldsymbol{\Psi}\boldsymbol{h}\|_2\leq\varepsilon
\label{compressed-sensing-formulation}
\end{align}
where $\varepsilon$ is an error tolerance parameter related to
noise statistics. Many efficient algorithms such as the fast
iterative shrinkage-thresholding algorithm (FISTA)
\cite{BeckTeboulle09} can be employed to solve the above sparse
signal recovery problem. Compressed sensing theory tells that, for
the noiseless case, we can perfectly recover a high-dimensional
sparse signal $\boldsymbol{h}$ from a much lower dimensional
linear measurement vector $\boldsymbol{y}$. Thus the compressed
sensing-based method has the potential to achieve a substantial
training overhead reduction.

%According to the compressed sensing theory, since the sparse
%vector $\boldsymbol{h}$ has $L$ non-zero entries, the required
%number of measurements $T$ to exactly recover $\boldsymbol{h}$
%from $\boldsymbol{y}$ is of order $O(L\log(n/L))\approx O(L)$,
%where $n=N_{\text{BS}} N_{\text{MS}}$.

%for each path in the $l$th cluster

\section{Channel Model with Angular Spreads}
\label{sec:channel-model} In addition to sparsity, mmWave channels
may exhibit angular spreads over the AoA, AoD, and elevation
domains \cite{RappaportSun13,AkdenizLiu14}. The angular spreads
are a result of scattering clusters, where each cluster may
contribute with multiple rays/paths with closely-spaced AoAs, AoDs
and elevations. To more accurately model the angular spread
characteristics of mmWave channels, we adopt the following
geometric channel model with $L$ clusters
\begin{align}
\boldsymbol{H} = \sum_{l=1}^L\bigg(\sum_{i=1}^{I}\alpha_{l,i}
\boldsymbol{a}_{\text{BS}}(\theta_{l}-\vartheta_{l,i})\bigg)
\bigg(\sum_{j=1}^{J}\beta_{l,j}\boldsymbol{a}_{\text{MS}}^H(\phi_{l}-\varphi_{l,j})\bigg)
\label{cluster-channel-model}
\end{align}
where each cluster has $IJ$ paths in total, $\theta_{l}$ and
$\phi_{l}$ represent the mean AoA/AoD associated with each
cluster, and $\vartheta_{l,i}$ and $\varphi_{l,j}$ denote the
relative AoA and AoD shift from the mean angle. Note that a
similar channel model was considered in \cite{AlkhateebHeath16},
where each cluster is assumed to contribute multiple rays/paths
between the BS and MS. In fact, the above model
(\ref{cluster-channel-model}) can be considered as a generalized
form of the channel model in \cite{AlkhateebHeath16}. On the other
hand, it can be easily observed that the above channel model can
also be expressed as a form of (\ref{channel-model}). Therefore
the compressed sensing-based channel estimation scheme
(\ref{compressed-sensing-formulation}) still applies.
Nevertheless, as to be shown in the following, the mmWave channel
with angular spreads not only exhibits sparsity patterns, it also
has a meaningful low-rank structure that can be simultaneously
utilized to reduce the sample complexity.

Similar to (\ref{virtual-channel}), we express the channel
(\ref{cluster-channel-model}) as a beam space MIMO representation
\begin{align}
\boldsymbol{H}=&\sum_{l=1}^L\boldsymbol{A}_{\text{BS}}\boldsymbol{\alpha}_l\boldsymbol{\beta}_l^T
\boldsymbol{A}_{\text{MS}}^H =
\boldsymbol{A}_{\text{BS}}\bigg(\sum_{l=1}^L\boldsymbol{\alpha}_l\boldsymbol{\beta}_l^T\bigg)\boldsymbol{A}_{\text{MS}}^H
\nonumber\\
\triangleq&
\boldsymbol{A}_{\text{BS}}\boldsymbol{H}_v\boldsymbol{A}_{\text{MS}}^H
\label{virtual-channel-2}
\end{align}
where $\boldsymbol{\alpha}_l\in\mathbb{C}^{N_1}$ and
$\boldsymbol{\beta}_l\in\mathbb{C}^{N_2}$ represent the virtual
representation over the AoA and AoD domain, respectively. Since
the angular spread occupies only a small portion of the whole
angular domain, both $\boldsymbol{\alpha}_l$ and
$\boldsymbol{\beta}_l$ are sparse vectors with only a few nonzero
entries concentrated around the mean AoA and AoD associated with
the $l$th cluster. Hence the virtual beam space channel
$\boldsymbol{H}_v$ is a sum of $L$ sparse matrices. Suppose any
sparse vector in
$\{\boldsymbol{\alpha}_l,\boldsymbol{\beta}_l\}_l$ contains at
most $p$ nonzero entries. As a result, $\boldsymbol{H}_v$ is a
sparse matrix with at most $p^2 L$ nonzero entries. Also,
$\boldsymbol{H}_v$ has at most $pL$ nonzero columns and at most
$pL$ nonzero rows. Note that due to the limited scattering nature
and small angular spreads, we usually have $pL\ll
\min\{N_1,N_2\}$. Meanwhile, $\boldsymbol{H}_v$ has a low rank
structure with $\text{rank}(\boldsymbol{H}_v)=L$. Thus the virtual
beam space channel has a simultaneously sparse and low-rank
structure.

%i.e. the minimum amount of training overhead

Our objective is to estimate/recover the joint sparse and low-rank
virtual channel $\boldsymbol{H}_v$ using as few measurements as
possible. Estimation of low-rank matrices or sparse matrices from
compressed linear measurements has been studied extensively in
various settings, e.g.
\cite{CandesTao05,TroppGilbert07,CandesRecht09,RechtFazel10,KoltchinskiiLounici11}.
However, there is much less research for cases where the matrix of
interest is characterized by two structures simultaneously. In
particular, how to simultaneously exploit both structures to
improve the sample complexity is of most concern. In
\cite{BahmaniRomberg15}, an efficient two-stage scheme was
developed for recovering a sparse, rank-one and positive
semi-definite matrix in the context of compressive phase
retrieval, and it was shown that the proposed two-stage scheme can
achieve a near-optimal sample complexity and enjoys nice
robustness guarantees. In the following section, the two-stage
scheme is extended to a more general scenario where the mmWave
channel to be estimated is not necessarily a rank-one positive
semi-definite matrix. We show that an reduced sample complexity
can be obtained as compared with simply exploiting the sparsity of
the mmWave channel.

\section{Two-Stage Compressed Sensing Scheme} \label{sec:proposed-method}
Before proceeding, we revisit the measurement collection model
(\ref{data-model}) and reformulate this measurement process as a
low-rank matrix sampling process. Assume $\boldsymbol{z}(t)$ and
$\boldsymbol{f}(t)$ are randomly chosen from pre-determined
beamforming/combining codebooks $\mathcal{Z}$ and $\mathcal{F}$,
respectively, where the cardinality of the two sets are
$|\mathcal{Z}|= N_Z$ and $|\mathcal{F}|=N_F$ and no beam pair
$\{\boldsymbol{z}(t),\boldsymbol{f}(t)\}$ is reused during the
sampling process. Let
$\boldsymbol{Z}\in\mathbb{C}^{N_{\text{BS}}\times N_Z}$ and
$\boldsymbol{F}\in\mathbb{C}^{N_{\text{MS}}\times N_F}$ be
matrices constructed by all vectors in $\mathcal{Z}$ and
$\mathcal{F}$, respectively. Then the observation model
(\ref{data-model}) can be expressed as sampling from a low-rank
matrix:
\begin{align}
\boldsymbol{Y}_{ij} =
(\boldsymbol{Z}^{H}\boldsymbol{H}\boldsymbol{F})_{ij}\quad
(i,j)\in\Omega \label{data-model-2}
\end{align}
where
$\boldsymbol{Y}\triangleq\boldsymbol{Z}^{H}\boldsymbol{H}\boldsymbol{F}$
is a low rank matrix with $\text{rank}(\boldsymbol{Y})=L$,
$\boldsymbol{Y}_{ij}$ denotes the $(i,j)$th entry of
$\boldsymbol{Y}$, and $\Omega$ denotes a set indicating which
entries of $\boldsymbol{Y}$ are observed. We have $|\Omega|=T$.
Also, here the observation noise is temporarily ignored to
simplify our subsequent analysis.

Suppose $\boldsymbol{Z}$ and $\boldsymbol{F}$ are full-rank square
matrices, i.e. $N_Z=N_{\text{BS}}$ and $N_F=N_{\text{MS}}$. Then
the problem of estimating $\boldsymbol{H}$ is equivalent to a
low-rank matrix completion problem. Specifically, we first recover
the low-rank matrix $\boldsymbol{Y}$ via a nuclear-norm
minimization \cite{CandesRecht09}:
\begin{align}
\min_{\boldsymbol{\hat{Y}}} &\quad  \|\boldsymbol{\hat{Y}}\|_{\ast} \nonumber\\
\text{s.t.}& \quad \boldsymbol{\hat{Y}}_{ij} = \boldsymbol{Y}_{ij}
\quad\forall (i,j)\in\Omega
\end{align}
After recovering $\boldsymbol{Y}$, the channel $\boldsymbol{H}$
can be estimated as
\begin{align}
\hat{\boldsymbol{H}}=(\boldsymbol{Z}^{H})^{-1}\boldsymbol{\hat{Y}}\boldsymbol{F}^{-1}
\end{align}
Nevertheless, according to the matrix completion theory
\cite{CandesRecht09}, the number of measurements has to satisfy
\begin{align}
T\geq Cn^{5/4}L\log(n)
\end{align}
in order to stably reconstruct $\boldsymbol{Y}$ of rank at most
$L$ with probability at least $1-c n^{-3}$, where
$n=\max\{N_\text{BS},N_\text{MS}\}$, and the constants $C,c>0$ are
universal. Hence for the low-rank matrix completion approach, the
required number of measurements is of order
$\mathcal{O}(L\max\{N_\text{BS},N_\text{MS}\}^{5/4})$, which
increases approximately linearly with the number of antennas
employed at the BS or MS, whichever is greater. We see that the
low-rank matrix completion scheme ignores the sparse structure
inherent in mmWave channels, and thus can only achieve a
sub-optimal sample complexity. To obtain a lower sample
complexity, we introduce the following two-stage compressed
sensing scheme.

\subsection{Proposed Scheme}
The idea of the proposed two-stage scheme is to exploit the low
rank and sparse structures in two separate stages. In the first
stage, we utilize the low rank structure to recover
$\boldsymbol{Y}$ from observations $\{\boldsymbol{Y}_{i,j},
(i,j)\in\Omega\}$. Note that $\boldsymbol{Z}$ and $\boldsymbol{F}$
do not need to be full-rank; instead, in order to achieve a lower
sample complexity, they should have reduced dimensions, i.e. $N_Z<
N_{\text{BS}}$ and $N_F< N_{\text{MS}}$. In other words, the size
of $\boldsymbol{Y}$ is much smaller than the size of
$\boldsymbol{H}$. In the second stage, based on the reconstructed
$\boldsymbol{Y}$, we estimate the virtual beam space channel
$\boldsymbol{H}_v$ by exploiting the sparse structure of
$\boldsymbol{H}_v$. Through this two-stage scheme, the low-rank
and sparse structures of the channel matrix $\boldsymbol{H}_v$ can
be effectively decoupled and thus better utilized. For clarity, we
summarize the two-stage scheme in Algorithm \ref{algorithm1}.

\begin{algorithm}
\caption{Two-Stage Compressed Sensing Algorithm}
\begin{algorithmic}
\STATE {Given the measurements $\boldsymbol{Y}_\Omega$, and the
matrices
$\boldsymbol{A}\triangleq\boldsymbol{Z}^H\boldsymbol{A}_{\text{BS}}$,
$\boldsymbol{B}\triangleq\boldsymbol{A}_{\text{MS}}^H\boldsymbol{F}$.
\begin{description}
  \item[1]
    Recover $\hat{\boldsymbol{Y}}$ by solving
    \begin{align}
    &\min_{\boldsymbol{\hat{Y}}} \|\boldsymbol{\hat{Y}}\|_* \nonumber\\
    &\text{s.t.} \quad \boldsymbol{\hat{Y}}_{ij} = \boldsymbol{Y}_{ij}
\quad\forall (i,j)\in\Omega
    \label{stage-1}
    \end{align}
  \item[2]
    Estimate $\hat{\boldsymbol{H}}_v$ via
    \begin{align}
    &\min_{\boldsymbol{H}_v} \|\boldsymbol{H}_v\|_1 \nonumber\\
    &\text{s.t.} \quad \hat{\boldsymbol{Y}} = \boldsymbol{A}^{H}
    \boldsymbol{H}_v\boldsymbol{B}
    \label{stage-2}
    \end{align}
\end{description}
}
\end{algorithmic}
\label{algorithm1}
\end{algorithm}

\subsection{Theoretical Results}
We now provide theoretical guarantees for our proposed two-stage
compressed sensing scheme. Our main results are summarized as
follows.

\newtheorem{theorem}{Theorem}
\begin{theorem} \label{theorem1}
Consider the channel estimation problem described in
(\ref{data-model-2}), where the indexes in $\Omega$ are uniformly
chosen at random with $|\Omega|= T$. The channel matrix
$\boldsymbol{H}$ can be represented in a form of
(\ref{virtual-channel-2}). Let $L$ denote the rank of
$\boldsymbol{H}$, and $p$ denote the maximum number of nonzero
entries in $\{\boldsymbol{\alpha}_l,\boldsymbol{\beta}_l\}_l$.
Suppose $\boldsymbol{A}\in \mathbb{C}^{N_Z\times N_1}$ and
$\boldsymbol{B}\in \mathbb{C}^{N_F\times N_2}$ are random matrices
with i.i.d. Gaussian random entries $a_{i,j}\sim
\mathcal{N}(0,\frac{1}{N_Z})$ and $b_{i,j}\sim
\mathcal{N}(0,\frac{1}{N_F})$\footnote{See discussions in Section
\ref{sec:proposed-method}.C regarding this assumption.}. Define
$n\triangleq\max\{N_F,N_Z\}$. There exist positive absolute
constants $c_1$, $c_2$, $c_3$, $c_4$, $c_5$ and $c_6$ such that if
\begin{align}
N_Z \geq & c_1 pL\log(N_\text{BS}/pL) \label{condition-1}\\
N_F \geq & c_2 pL\log(N_\text{MS}/pL) \label{condition-2}\\
T\geq & c_3 n^{5/4}L\log(n) \label{condition-3}
\end{align}
then the channel $\boldsymbol{H}$ can be perfectly recovered from
Algorithm \ref{algorithm1} with probability exceeding $(1-c_4
n^{-3})(1-2e^{-c_5N_Z})(1-2e^{-c_6N_F})$.
\end{theorem}

\begin{proof}
Our proof proceeds in two steps. We first investigate the
condition under which $\boldsymbol{Y}$ can be perfectly recovered
from (\ref{stage-1}), and then examine the exact recovery
condition for (\ref{stage-2}). By combining the results of the two
stages, we arrive at results in Theorem \ref{theorem1}.

Since $\boldsymbol{Y}$ has a low rank structure, the first stage
is essentially a matrix completion stage. Invoking the matrix
completion theory \cite{CandesRecht09}, we know that for some
positive constants $c_3$ and $c_4$, if (\ref{condition-3}) holds,
then $\boldsymbol{Y}$ can be perfectly recovered with probability
exceeding $1-c_4n^{-3}$.

The second stage is a sparse matrix recovery stage. Note that
$\boldsymbol{H}_v$ is a sparse matrix with at most $pL$ nonzero
columns and rows. We have the following theoretical guarantee for
recovering a sparse matrix $\boldsymbol{X}$ from compressed linear
measurements
$\boldsymbol{G}=\boldsymbol{A}\boldsymbol{X}\boldsymbol{B}$.
\newtheorem{lemma}{Lemma}
\begin{lemma}
Let $\boldsymbol{X}\in \mathbb{C}^{N_1\times N_2}$ denote a sparse
matrix with at most $k$ nonzero columns and rows.
$\boldsymbol{A}\in \mathbb{C}^{N_A\times N_1}$ and
$\boldsymbol{B}\in \mathbb{C}^{N_B\times N_2}$ satisfy the
$2k$-restricted isometry property with $\delta_{2k}$, namely,
\begin{align}
(1-\delta_{2k})\left\|\boldsymbol{x}\right\|_2^2 \leq
\left\|\boldsymbol{A}\boldsymbol{x}\right\|_2^2
\leq (1+\delta_{2k})\left\|\boldsymbol{x}\right\|_2^2 \nonumber\\
(1-\delta_{2k})\left\|\boldsymbol{x}\right\|_2^2 \leq
\left\|\boldsymbol{B}\boldsymbol{x}\right\|_2^2 \leq
(1+\delta_{2k})\left\|\boldsymbol{x}\right\|_2^2 \nonumber
\end{align}
for all $2k$-sparse vectors $\boldsymbol{x}$, where $\delta_{2k}
\triangleq\max\{\delta_{2k}(\boldsymbol{A}),\delta_{2k}(\boldsymbol{B})\}$,
with $\delta_{2k}(\boldsymbol{A})$ and
$\delta_{2k}(\boldsymbol{B})$ denoting the restricted isometry
constants (RIC) of $\boldsymbol{A}$ and $\boldsymbol{B}$
respectively. If the following condition holds
\begin{align}
\delta_{2k}<1+\sqrt{2}\left(1-\sqrt{1+\sqrt{2}}\right)\approx
0.216
\end{align}
then $\boldsymbol{X}$ can be exactly recovered via
\begin{align}
\min_{\boldsymbol{\hat{X}}} &\quad \|\boldsymbol{\hat{X}}\|_1 \nonumber\\
\text{s.t.} &\quad \boldsymbol{G} =
\boldsymbol{A}\boldsymbol{\hat{X}}\boldsymbol{B}^H \label{opt1}
\end{align}
\label{lemma1}
\end{lemma}
\begin{proof}
See Appendix \ref{appA}.
\end{proof}

Meanwhile, it is well-known that for a random matrix
$\boldsymbol{\Psi}\in\mathbb{R}^{m\times n}$ whose i.i.d. entries
follow a Gaussian distribution with zero mean and variance $1/m$,
if the following condition
\begin{align}
m\geq \eta k\log (n/k)
\end{align}
holds for a sufficiently large constant $\eta>0$, then
$\boldsymbol{\Psi}$ satisfies the $2k$-restricted isometry
property for a sufficiently small restricted isometry constant
$\delta_{2k}(\boldsymbol{\Psi})$ with probability exceeding
$1-2e^{-cm}$ for some constant $c>0$ that depends only on
$\delta_{2k}(\boldsymbol{\Psi})$ \cite{BaraniukDavenport08}.
Recalling Lemma \ref{lemma1}, we therefore can naturally arrive at
the following: for some positive constants $c_1$, $c_2$, $c_5$ and
$c_6$, if (\ref{condition-1}) and (\ref{condition-2}) hold valid,
then $\boldsymbol{H}_v$ can be perfectly recovered via
(\ref{stage-2}) with probability exceeding
$(1-2e^{-c_5N_Z})(1-2e^{-c_6N_F})$.

By combining the results from both stages, we now reach that there
exist positive absolute constants $c_1$, $c_2$, $c_3$, $c_4$,
$c_5$ and $c_6$ such that if
(\ref{condition-1})--(\ref{condition-3}) are satisfied, then the
channel $\boldsymbol{H}$ can be perfectly recovered from Algorithm
\ref{algorithm1} with probability exceeding $(1-c_4
n^{-3})(1-2e^{-c_5N_Z})(1-2e^{-c_6N_F})$. The proof is completed
here.
\end{proof}

%We have the following remarks regarding Theorem \ref{theorem1}.

%and logarithmically with the number of antennas at the BS or MS,
%whichever is greater.

%Nevertheless, due to the limited scattering characteristic of
%mmWave channels, $L$ is usually smaller than the number of
%antennas at the

\subsection{Discussions}
From Theorem \ref{theorem1}, we see that the number of
measurements $T$ required for exact channel recovery is of order
\begin{align}
\mathcal{O}(p^{5/4}L^{9/4}\log(n))\approx\mathcal{O}(pL^{2})
\label{sample-complexity-two-stage}
\end{align}
which scales approximately linearly with $p$ and quadratically
with the rank $L$. Since $p$ and $L$ are usually much smaller than
$\max\{N_{\text{BS}},N_{\text{MS}}\}$, our proposed two-stage
scheme can achieve substantial overhead reduction as compared with
the low rank matrix completion scheme whose required number of
measurements scales linearly with
$\max\{N_{\text{BS}},N_{\text{MS}}\}$.

It is also interesting to compare our proposed two-stage scheme
with a compressed sensing method which solves (\ref{data-model-2})
by directly formulating (\ref{data-model-2}) into a sparse
recovery problem (\ref{compressed-sensing-formulation}). Note that
$\boldsymbol{h}=\text{vec}(\boldsymbol{H}_v)$ has at most $p^2L$
nonzero entries. According to the compressed sensing theory
\cite{CandesTao05}, we know that the probability of successful
recovery of $\boldsymbol{h}$ via
(\ref{compressed-sensing-formulation}) exceeds $1-\delta$ if
\begin{align}
T\geq Cp^2L\log (N_{1}N_{2}/\delta)
\end{align}
in which $C$ is a positive constant. Thus the number of
measurements required for exact channel recovery is of order
\begin{align}
\mathcal{O}(p^2 L) \label{sample-complexity-CS}
\end{align}
for the direct compressed sensing method. Comparing
(\ref{sample-complexity-two-stage}) with
(\ref{sample-complexity-CS}), we can see that our proposed
two-stage scheme achieves a lower sample complexity than the
direct compressed sensing method if $L<p$. Note that $L$
represents the number of scattering clusters, and $p$, the largest
number of nonzero entries in
$\{\boldsymbol{\alpha}_l,\boldsymbol{\beta}_l\}$, is a quantity
that measures the maximum angular spread among all scattering
clusters. Due to the limited scattering characteristics in mmWave
channels, we usually have $L<p$ in practice. In particular, for
the extreme case where there is only a line-of-sight (LOS) path
between the transmitter and the receiver, $L$ is equal to one,
whereas $p$ is generally greater than one since there still exists
angular spread (power leakage) due to limited spatial resolution.

%Therefore whether our proposed scheme or the direct compressed
%sensing method is more advantageous depends on which one among $L$
%and $p$ is larger

In Theorem \ref{theorem1}, we assume that
$\boldsymbol{A}\triangleq\boldsymbol{Z}^H\boldsymbol{A}_{\text{BS}}$
and
$\boldsymbol{B}\triangleq\boldsymbol{A}_{\text{MS}}^H\boldsymbol{F}$
are random matrices with i.i.d. Gaussian random entries.
Nevertheless, noticing that $\boldsymbol{A}_{\text{BS}}$ and
$\boldsymbol{A}_{\text{MS}}$ are structured matrices consisting of
array response vectors, it may not be possible to devise
beamforming and combining matrices
$\{\boldsymbol{Z},\boldsymbol{F}\}$ such that the resulting
$\boldsymbol{A}$ and $\boldsymbol{B}$ satisfy the i.i.d. Gaussian
assumption. We, however, still make such an assumption in order to
facilitate our theoretical analysis. On the other hand, recent
theoretical and empirical studies \cite{DuarteEldar11} show that
structured matrices also enjoy nice restricted isometry
properties. Note that the same problem exists for the direct
compressed sensing method, where the sensing matrix is highly
structured but a random sensing matrix assumption is evoked in
order to obtain its sample complexity.

\subsection{Extension To The Noisy Case}
In the previous subsections, we ignore the observation noise in
order to simplify our theoretical analysis. Nevertheless, the
two-stage compressed sensing scheme can be easily adapted to the
noisy case. For clarity, the two-stage algorithm for the noisy
case is summarized as follows.

\begin{algorithm}
\caption{Robust Two-Stage Compressed Sensing Algorithm}
\begin{algorithmic}
\STATE {Given the measurements $\boldsymbol{Y}_\Omega$, the
matrices
$\boldsymbol{A}\triangleq\boldsymbol{Z}^H\boldsymbol{A}_{\text{BS}}$,
$\boldsymbol{B}\triangleq\boldsymbol{A}_{\text{MS}}^H\boldsymbol{F}$.
\begin{description}
  \item[1]
    Recover $\hat{\boldsymbol{Y}}$ by solving
    \begin{align}
    &\text{min}\ \|\boldsymbol{\hat{Y}}\|_* \nonumber\\
    &\text{s.t.} \quad \|\hat{\boldsymbol{Y}}_\Omega - \boldsymbol{Y}_\Omega\|_F < \varepsilon
    \label{stage-1-noisy}
    \end{align}
  \item[2]
    Estimate $\hat{\boldsymbol{H}}_v$ via
    \begin{align}
    &\text{min}\ \|\boldsymbol{H}_v\|_1 \nonumber\\
    &\text{s.t.} \quad \|\hat{\boldsymbol{Y}} - \boldsymbol{A}^{H}
    \boldsymbol{H}_v\boldsymbol{B}\|_F < \epsilon
    \label{stage-2-noisy}
    \end{align}
\end{description}
}
\end{algorithmic}
\label{algorithm2}
\end{algorithm}

In Algorithm \ref{algorithm2}, $\varepsilon$ and $\epsilon$ are
error tolerance parameters. Also, the constrained optimizations
(\ref{stage-1-noisy}) and (\ref{stage-2-noisy}) can be converted
to unconstrained optimization problems by introducing an
appropriate choice of the regularization parameter $\lambda$. For
example, (\ref{stage-1-noisy}) can be replaced by
\begin{align}
\min_{\boldsymbol{\hat{Y}}}\quad \|\hat{\boldsymbol{Y}}_\Omega -
\boldsymbol{Y}_\Omega\|_F^2+\lambda \|\boldsymbol{\hat{Y}}\|_{*}
\end{align}
which can be efficiently solved by the fixed point continuation
algorithm \cite{MaGoldfarb11}.

\begin{figure*}[!t]
\centering
 \subfigure[Success rates vs. $T$.]{\includegraphics[width=3.5in]{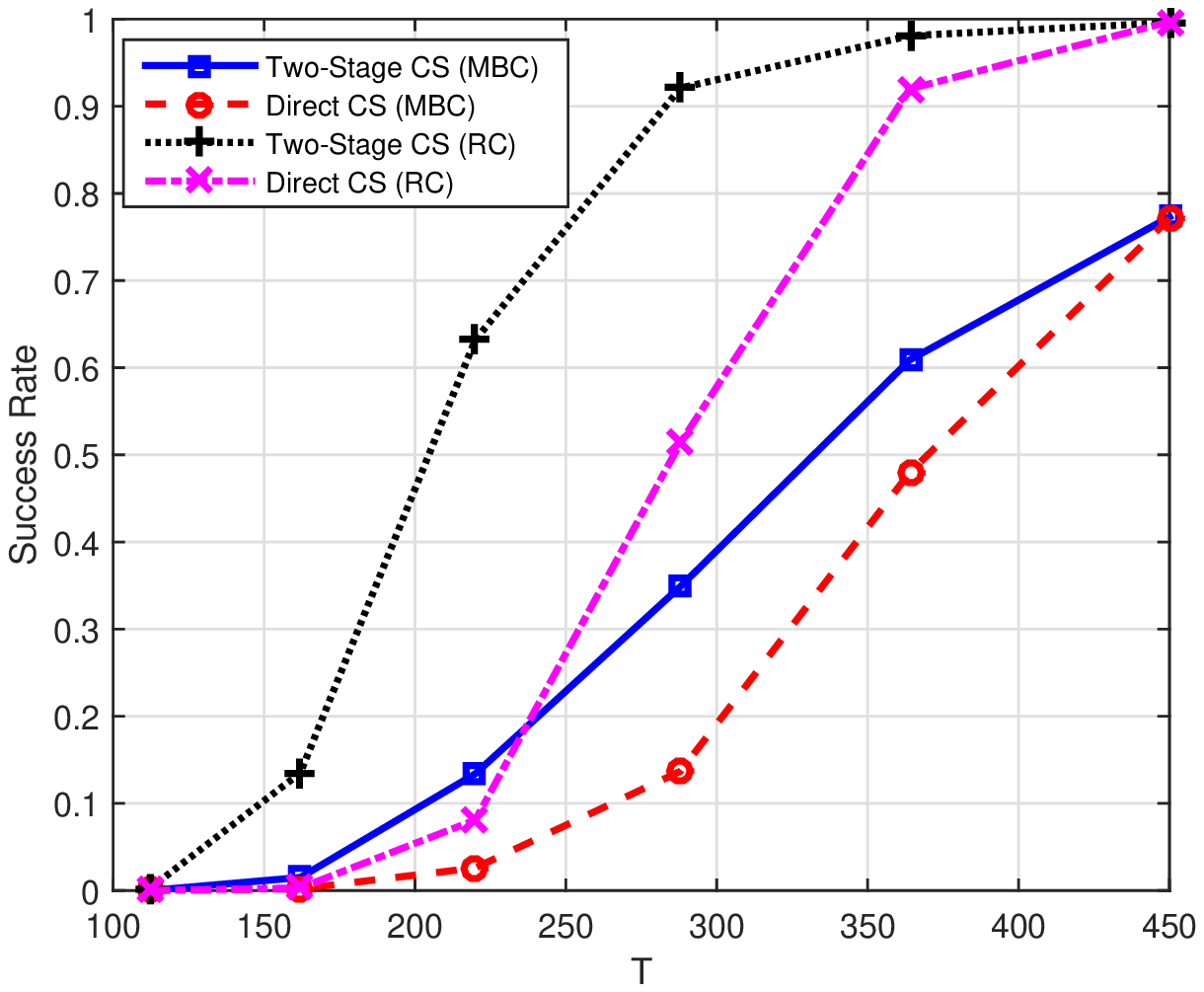}} \hfil
\subfigure[NMSEs vs. $T$.]{\includegraphics[width=3.5in]{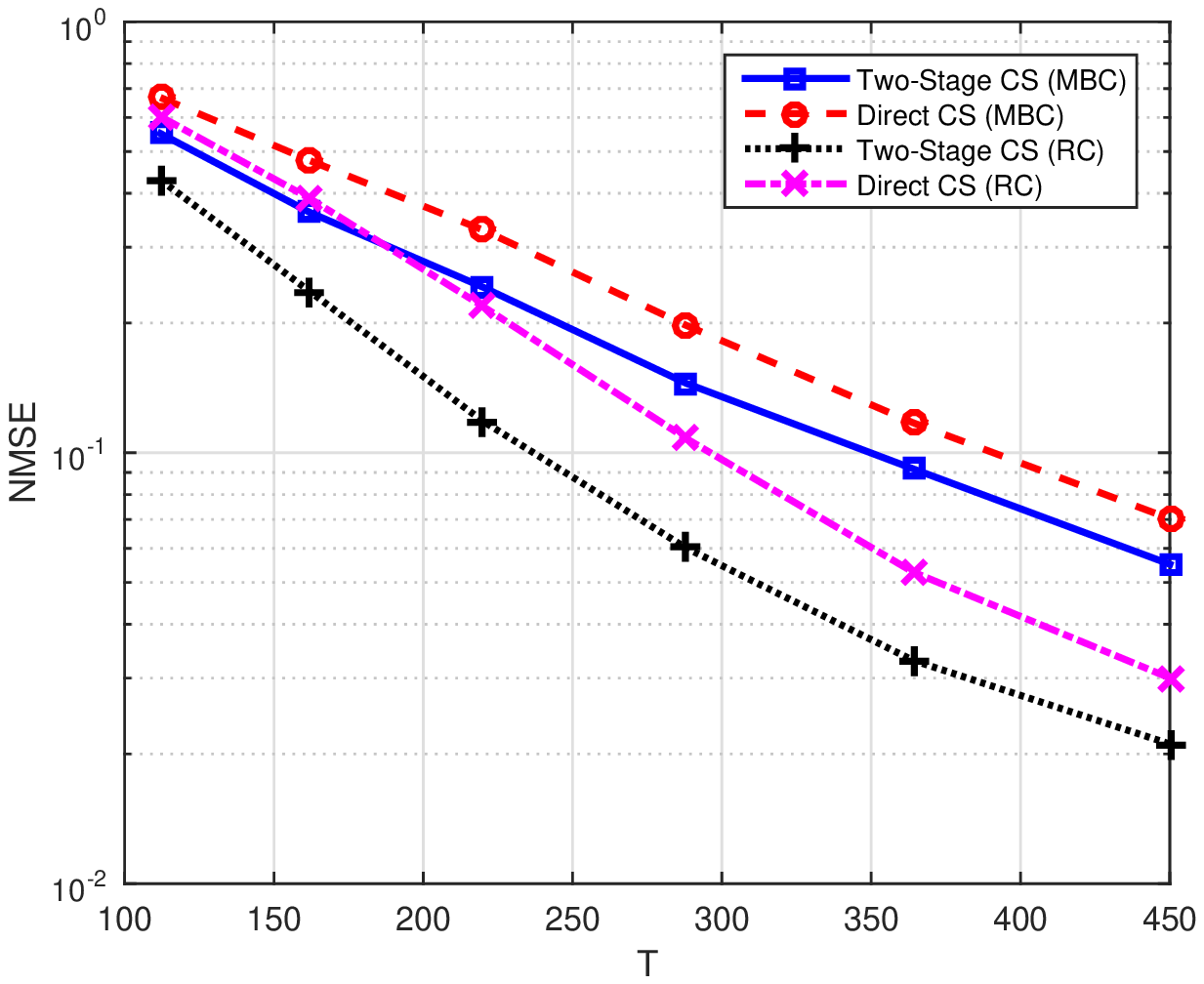}}
\caption{Success rates and NMSEs of respective algorithms vs.
$T$.} \label{fig2}
\end{figure*}

\begin{figure*}[!t]
\centering
 \subfigure[Success rates vs. the angular spread.]{\includegraphics[width=3.5in]{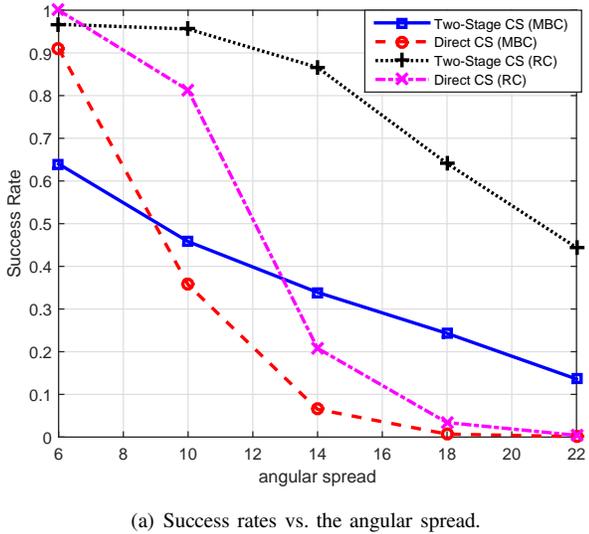}} \hfil
\subfigure[NMSEs vs. the angular
spread.]{\includegraphics[width=3.5in]{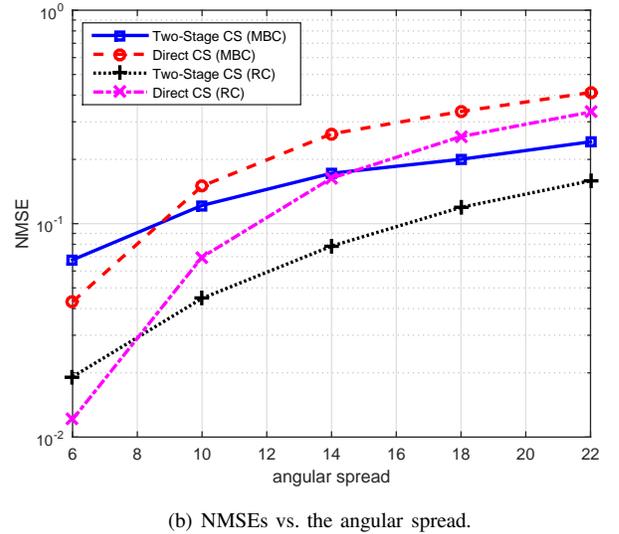}}
\caption{Success rates and NMSEs of respective algorithms vs. the
angular spread.} \label{fig3}
\end{figure*}

\section{Simulation Results} \label{sec:experiments}
We now carry out simulation results to illustrate the performance
of our proposed two-stage compressed sensing (referred to as
two-stage CS) method and its comparison with the direct compressed
sensing (referred to as direct-CS) method. For our proposed
method, we use the singular value thresholding (SVT) algorithm
\cite{CaiCandes10} and the fixed point continuation (FPC)
algorithm \cite{MaGoldfarb11} to solve the matrix completion
problem for the noiseless and noisy case, respectively. A fast
iterative shrinkage-thresholding algorithm (FISTA)
\cite{BeckTeboulle09} is employed to perform the sparse recovery
stage and to solve the direct CS method.

%It can be easily verified that the AoA/AoD angular spreads in
%terms of rms (i.e. the standard deviations of AoA and AoD) are
%equivalent to $4.3^\circ$ and $2.9^\circ$, respectively, which are
%smaller than the reported results
%\cite{RappaportSun13,AkdenizLiu14} in dense-urban propagation
%environments.

We consider a scenario where both the BS and the MS employ a
uniform linear array with $N_{\text{BS}}=N_{\text{MS}}=64$
antennas. The distance between neighboring antenna elements is
assumed to be half the wavelength of the signal. The mmWave
channel is assumed to follow the geometric channel model
(\ref{cluster-channel-model}) with $L=2$ clusters. The mean
AoAs/AoDs for these two clusters are set to
$\theta_1=\phi_1=\pi/6$, $\theta_2=\phi_2=-\pi/6$, respectively.
The number of rays within each cluster is set to $IJ=100$. Unless
otherwise specified, the AoA and AoD angular spreads for each
cluster are set to $\delta_{\theta}=15^\circ$ and
$\delta_{\phi}=10^\circ$. The relative AoA/AoD shifts are
uniformly generated within the angular spreads, i.e.
$\vartheta_{l,i}\in
(\theta_l-\delta_{\theta}/2,\theta_l+\delta_{\theta}/2)$,
$\varphi_{l,i}\in
(\phi_l-\delta_{\phi}/2,\phi_l+\delta_{\phi}/2)$. The complex
gains $\{\alpha_{l,i}\beta_{l,j}\}$ are assumed to be random
variables following a circularly symmetric complex Gaussian
distribution $\mathcal{CN}(0,1/\rho)$, where $\rho$ is given by
$\rho=(4\pi Df_c/c)^2$. Here $c$ represents the speed of light,
$D$ denotes the distance between the BS and the MS, $f_c$ is the
carrier frequency, and we set $D=30\text{m}$ and
$f_c=28\text{GHz}$. The performance is evaluated via two metrics,
namely, the normalized mean squared error (NMSE) and the success
rate. The NMSE is calculated as
\begin{align}
\text{NMSE}=E\left[\frac{\left\|\hat{\boldsymbol{H}}
-\boldsymbol{H}\right\|^2_F}{\left\|\boldsymbol{H}\right\|^2_F}\right]
\end{align}
where $\boldsymbol{\hat{H}}$ denotes the estimate of the true
channel $\boldsymbol{H}$. The success rate is computed as the
ratio of the number of successful trials to the total number of
independent runs. A trial is considered successful if the
normalized reconstruction error is no greater than $10^{-2}$.

In our experiments, the beamforming/combining codebooks, i.e.
$\boldsymbol{F}$ and $\boldsymbol{Z}$, are generated according to
two different ways. The first is to have the entries of
$\boldsymbol{F}$ and $\boldsymbol{Z}$ uniformly chosen from a unit
circle, in which case the antenna array has a
quasi-omnidirectional beam pattern. This scheme is referred to as
a random coding (RC) scheme. Another scheme of devising
$\boldsymbol{F}$ and $\boldsymbol{Z}$ is to steer the antenna
array to beam along multiple directions simultaneously, which is
achieved by dividing the antenna array into a number of sub-arrays
and making each sub-array beam toward an individual direction
\cite{AbariHassanieh16}. The steering directions are randomized
for each measurement. This scheme is named as multiple-beam coding
(MBC) scheme. In order to provide a fair comparison, the columns
of $\boldsymbol{F}$ and $\boldsymbol{Z}$ are normalized to unit
norm for both beam pattern design schemes. We assume that, at each
time instant, the beamforming vector $\boldsymbol{f}(t)$ and the
combining vector $\boldsymbol{z}(t)$ are randomly chosen from the
beamforming/combining codebooks, respectively. Hence the
measurement process can be deemed as randomly collecting samples
from a low-rank matrix
$\boldsymbol{Y}=\boldsymbol{Z}^H\boldsymbol{H}\boldsymbol{F}$ (cf.
(\ref{data-model-2})), where $\boldsymbol{Y}$ is an $N_Z\times
N_F$ matrix. For simplicity, we assume $N_Z=N_F$. Also, in our
experiments, the value of $N_Z$ ($N_F$) is adaptively adjusted
such that the ratio of the number of observed entries $T$ to the
total number of entries in $\boldsymbol{Y}$ is fixed to be $1/2$,
i.e. $T=(1/2)N_{Z} N_{F}$. Such a setup can provide a reliable
matrix completion result, which in turn helps achieve an accurate
channel estimate for our proposed two-stage method. The adaptive
adjustment of the dimensions of the codebooks can be easily
implemented in practice. We can first generate augmented
beamforming/combining codebooks and then choose $\boldsymbol{Z}$
and $\boldsymbol{F}$ as subsets (with variable dimensions) of the
augmented codebooks.

We now examine the estimation performance of our proposed
two-stage CS method and the direct CS method. Fig. \ref{fig2}
plots the success rates for the noiseless case and NMSEs for the
noisy case as a function of the number of measurements $T$, where
for the noisy case, the SNR, defined as
$10\log(\|\boldsymbol{H}\|_F^2/(N_{\text{BS}}N_{\text{MS}}\sigma^2))$,
is set equal to 20dB. From Fig. \ref{fig2}, we see that better
performance can be obtained by using the beamforming/combining
codebooks that are generated according to the RC scheme. Also, our
proposed two-stage CS method presents a clear performance
advantage over the direct CS algorithm, whichever
beamforming/combining codebooks are used. This result corroborates
our claim that the proposed two-stage CS method can achieve a
lower sample complexity than the direct CS method.

%the AoA/AoD angular spreads in terms of rms range from $1.7^\circ$
%to $6.4^\circ$

Next, in Fig. \ref{fig3}, we examine the performance of respective
algorithms as a function of the angular spread, where the AoA and
AoD angular spreads are assumed to be the same and vary from
$6^\circ$ to $22^\circ$, i.e. $\delta_{\theta}=\delta_{\phi}\in
[6^\circ, 22^\circ]$. Also, we set $N_Z=N_F=24$, $T=0.5N_Z N_F$,
and the SNR is set to 20dB for the noisy case. From Fig.
\ref{fig3}, we see that the direct CS method outperforms our
proposed two-stage scheme when the angular spread is small, say,
$\delta_{\theta}=\delta_{\phi}=6^\circ$, whereas our proposed
method achieves a performance improvement over the direct CS as
the angular spread becomes large. This result, again,
substantiates our theoretical analysis. As indicated earlier in
our paper, our proposed two-stage scheme achieves a lower sample
complexity only when $L<p$, where $L$ represents the number of
scattering clusters, and $p$ is a value related to the angular
spread. When the angular spread is small, the condition $L<p$ may
not hold. As a result, the proposed two-stage CS method does not
necessarily perform better than the direct CS method. Lastly, in
Fig. \ref{fig4}, we depict the NMSEs of respective algorithms vs.
the SNR, where we set $N_Z=N_F=24$, $T=0.5N_ZN_F$,
$\delta_{\theta}=15^\circ$ and $\delta_{\phi}=10^\circ$. We see
that the proposed two-stage CS method outperforms the direct CS
method in moderate and high SNR regimes.

%The result, again, demonstrates the superiority of our proposed
%two-stage algorithm over the direct CS method.

\begin{figure}[!t]
\centering
\includegraphics[width=3.5in]{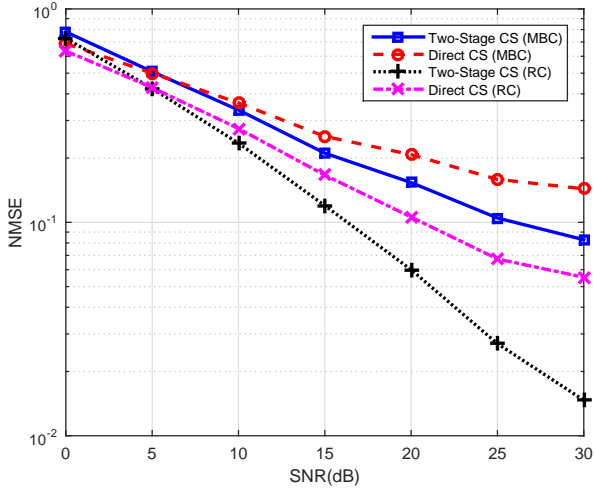}
\caption{NMSEs of respective algorithms vs. SNR.} \label{fig4}
\end{figure}

\section{Conclusions} \label{sec:conclusion}
We studied the problem of channel estimation for mmWave systems
with only one RF chain used at the BS and MS. Besides the sparse
scattering characteristics, we also considered the effect of
angular spreads in channel modeling and algorithm development. We
showed that, in the presence of angular spreads, mmWave channels
exhibit a jointly sparse and low-rank structure. A two-stage
compressed sensing method was developed, in which a matrix
completion stage is first performed, and then followed by a sparse
recovery stage to estimate the mmWave channel. Theoretical
analysis was also conducted. It reveals that the proposed
two-stage method requires fewer measurements than a direct
compressed sensing method that exploits the sparsity but ignore
the low-rank structure of mmWave channels. Simulation results were
provided to corroborate our theoretical analysis and demonstrate
the superiority of the proposed two-stage compressed sensing
method.

\useRomanappendicesfalse
\appendices

\section{Proof of Lemma \ref{lemma1}} \label{appA}
Before proving $\hat{\boldsymbol{X}}=\boldsymbol{X}$, we first
show that for any sparse matrix
$\boldsymbol{\Phi}\in\mathbb{C}^{N_1\times N_2}$ with at most $2k$
nonzero columns and rows, we have
\begin{align}
(1-\delta_{2k})^2\left\|\boldsymbol{\Phi}\right\|_F^2 \leq
\left\|\boldsymbol{A}\boldsymbol{\Phi}\boldsymbol{B}^H\right\|_F^2
\leq (1+\delta_{2k})^2\left\|\boldsymbol{\Phi}\right\|_F^2
\label{lemma1:RIP}
\end{align}
Since $\boldsymbol{A}$ satisfies the $2k$-RIP and each column of
$\boldsymbol{\Phi}$ is a $2k$-sparse vector, adding all the
inequalities together leads to
\begin{align}
(1-\delta_{2k})\left\|\boldsymbol{\Phi}\right\|_F^2 \leq
\left\|\boldsymbol{A}\boldsymbol{\Phi}\right\|_F^2 \leq
(1+\delta_{2k})\left\|\boldsymbol{\Phi}\right\|_F^2
\end{align}
Meanwhile, note that $\boldsymbol{\Phi}^H\boldsymbol{A}^H$ has at
most $2k$ non-zero rows, i.e. each column of
$\boldsymbol{\Phi}^H\boldsymbol{A}^H$ is also a $2k$-sparse
vector. Using the RIP associated with $\boldsymbol{B}$, we have
\begin{align}
\left\|\boldsymbol{B}\boldsymbol{\Phi}^H\boldsymbol{A}^H\right\|_F^2
\leq (1+\delta_{2k})
\left\|\boldsymbol{\Phi}^H\boldsymbol{A}^H\right\|_F^2 \leq (1+\delta_{2k})^2\left\|\boldsymbol{\Phi}\right\|_F^2
\label{lemma1-eqn1}\\
\left\|\boldsymbol{B}\boldsymbol{\Phi}^H\boldsymbol{A}^H\right\|_F^2
\geq
(1-\delta_{2k})\left\|\boldsymbol{\Phi}^H\boldsymbol{A}^H\right\|_F^2
\geq (1-\delta_{2k})^2\left\|\boldsymbol{\Phi}\right\|_F^2
\label{lemma1-eqn2}
\end{align}
Combining (\ref{lemma1-eqn1})--(\ref{lemma1-eqn2}), we arrive at
(\ref{lemma1:RIP}).

%(except for the last few ones)

Using (\ref{lemma1:RIP}), we now prove that $\boldsymbol{E}
\triangleq \hat{\boldsymbol{X}}-\boldsymbol{X}$ equals to zero,
i.e. $\|\boldsymbol{E}\|_F=0$. Let $\Omega$ denotes the support
set (i.e. the set of indices of non-zeros entries) of
$\boldsymbol{X}$. $\boldsymbol{E}$ can be decomposed as
\begin{align}
\boldsymbol{E} = \sum_{i=0}^N\boldsymbol{E}_i
\end{align}
where $\boldsymbol{E}_0$ is a matrix whose entries in the set
$\Omega$ are equivalent to those of $\boldsymbol{E}$, while the
rest of entries are equal to zero, $\boldsymbol{E}_i$ ($i\neq0$)
have disjoint support sets with size $k\times k$ such that
$(1/k^2)\|\boldsymbol{E}_i\|_1\geq
\|\boldsymbol{E}_{i+1}\|_{\infty}$ for $i=1,\ldots,N-1$. Note that
this inequality can be automatically satisfied if we arrange the
entries of $\boldsymbol{E}$ in descending order and ensure that
the largest (in terms of magnitude) entry in
$\boldsymbol{E}_{i+1}$ is no greater than the smallest entry in
$\boldsymbol{E}_{i}$. Since $\hat{\boldsymbol{X}}$ is an optimal
solution to (\ref{opt1}), we have
\begin{align}
\|\boldsymbol{X}\|_1 \geq \|\hat{\boldsymbol{X}}\|_1 &= \|\boldsymbol{E} + \boldsymbol{X} -
\boldsymbol{E}_0 + \boldsymbol{E}_0\|_1 \nonumber\\
&\geq \|\boldsymbol{E} + \boldsymbol{X} - \boldsymbol{E}_0\|_1 - \|\boldsymbol{E}_0\|_1 \nonumber\\
& = \|\boldsymbol{X}\|_1 + \|\boldsymbol{E} - \boldsymbol{E}_0\|_1
- \|\boldsymbol{E}_0\|_1
\end{align}
Thus we obtain
\begin{align}
\|\boldsymbol{E} - \boldsymbol{E}_0\|_1 \leq
\|\boldsymbol{E}_0\|_1 \overset{(a)}{\leq} k\|\boldsymbol{E}_0\|_F
\label{lemma1:eqn3}
\end{align}
where $(a)$ comes from the Cauchy-Schwarz inequality. Also, we
have
\begin{align}
\|\boldsymbol{E}- (\boldsymbol{E}_0+\boldsymbol{E}_1)\|_F &=
\sum_{i=2}^N \|\boldsymbol{E}_i\|_F \overset{(a)}{\leq}
\frac{1}{k}\sum_{i=1}^{N-1}\|\boldsymbol{E}_{i}\|_1
\overset{(b)}{\leq} \frac{1}{k}\|\boldsymbol{E}_0\|_1 \nonumber \\
&\overset{(c)}{\leq} \|\boldsymbol{E}_0\|_F \leq
\|\boldsymbol{E}_0+\boldsymbol{E}_1\|_F \label{lemma1:eqn4}
\end{align}
where $(a)$ comes from the fact that
\begin{align}
\|\boldsymbol{E}_{i}\|_1\geq
k^2\|\boldsymbol{E}_{i+1}\|_{\infty}\geq k
\|\boldsymbol{E}_{i+1}\|_{F}
\end{align}
and the inequalities $(b)$ and $(c)$ follow from
(\ref{lemma1:eqn3}). The result (\ref{lemma1:eqn4}) implies that
\begin{align}
\|\boldsymbol{E}\|_F \leq 2\|\boldsymbol{E}_0+\boldsymbol{E}_1\|_F
\end{align}

%by utilizing the disjoint support property and the inequality
%$(1/k^2)\|\boldsymbol{E}_i\|_1\geq
%\|\boldsymbol{E}_{i+1}\|_{\infty}$,

We now prove $\|\boldsymbol{E}_0+\boldsymbol{E}_1\|_F=0$. Note
that $\boldsymbol{E}_0+\boldsymbol{E}_1$ is a sparse matrix with
at most $2k$ nonzero columns and rows. Using (\ref{lemma1:RIP}),
we have
\begin{align}
&(1-\delta_{2k})^2\left\|\boldsymbol{E}_0+\boldsymbol{E}_1\right\|_F^2
\leq \left\|\boldsymbol{A}(\boldsymbol{E}_0+\boldsymbol{E}_1)\boldsymbol{B}^H\right\|_F^2 \nonumber \\
=&\text{tr}[(\boldsymbol{A}(\boldsymbol{E}_0+\boldsymbol{E}_1)
\boldsymbol{B}^H)^H\boldsymbol{A}(\boldsymbol{E}-\sum_{i=2}^N\boldsymbol{E}_i)\boldsymbol{B}^H] \nonumber \\
=&\Re\{\text{tr}[(\boldsymbol{A}(\boldsymbol{E}_0+\boldsymbol{E}_1)\boldsymbol{B}^H)^H
\boldsymbol{A}\boldsymbol{E}\boldsymbol{B}^H]\} \nonumber \\
&-\Re\{\text{tr}[(\boldsymbol{A}(\boldsymbol{E}_0+\boldsymbol{E}_1)
\boldsymbol{B}^H)^H\boldsymbol{A}\sum_{i=2}^N\boldsymbol{E}_i\boldsymbol{B}^H]\} \nonumber \\
\leq&\Re\{\text{tr}[(\boldsymbol{A}(\boldsymbol{E}_0+\boldsymbol{E}_1)
\boldsymbol{B}^H)^H\boldsymbol{A}\boldsymbol{E}\boldsymbol{B}^H]\} \nonumber \\
&+\left|\Re\{\text{tr}[(\boldsymbol{A}(\boldsymbol{E}_0+\boldsymbol{E}_1)
\boldsymbol{B}^H)^H\boldsymbol{A}\sum_{i=2}^N\boldsymbol{E}_i\boldsymbol{B}^H]\}\right| \nonumber \\
\overset{(a)}{\leq}&\sum_{i=0}^1\sum_{j=2}^N\left|\Re\{\text{tr}((\boldsymbol{A}
\boldsymbol{E}_i\boldsymbol{B}^H)^H\boldsymbol{A}\boldsymbol{E}_j\boldsymbol{B}^H)\}\right|  \nonumber \\
=&\sum_{i=0}^1\sum_{j=2}^N\left|\Re\{\text{tr}((\boldsymbol{A}
\frac{\boldsymbol{E}_i}{\|\boldsymbol{E}_i\|_F}\boldsymbol{B}^H)^H\boldsymbol{A}
\frac{\boldsymbol{E}_j}{\|\boldsymbol{E}_j\|_F}\boldsymbol{B}^H)\}\right|\cdot
\|\boldsymbol{E}_i\|_F\|\boldsymbol{E}_j\|_F \nonumber \\
\overset{(b)}{=}&\sum_{i=0}^1\sum_{j=2}^N\frac{1}{4}\bigg|\|\boldsymbol{A}(\frac{\boldsymbol{E}_i}{\|\boldsymbol{E}_i\|_F}
+\frac{\boldsymbol{E}_j}{\|\boldsymbol{E}_j\|_F})\boldsymbol{B}^H\|_F^2
\nonumber \\
&-\|\boldsymbol{A}(\frac{\boldsymbol{E}_i}{\|\boldsymbol{E}_i\|_F}-
\frac{\boldsymbol{E}_j}{\|\boldsymbol{E}_j\|_F})\boldsymbol{B}^H\|_F^2\bigg|
\cdot \|\boldsymbol{E}_i\|_F\|\boldsymbol{E}_j\|_F \nonumber\\
\leq&\sum_{i=0}^1\sum_{j=2}^N\frac{1}{4}\bigg((1+\delta_{2k})^2\|\frac{\boldsymbol{E}_i}{\|\boldsymbol{E}_i\|_F}+
\frac{\boldsymbol{E}_j}{\|\boldsymbol{E}_j\|_F}\|_F^2
\nonumber \\
&-(1-\delta_{2k})^2\|\frac{\boldsymbol{E}_i}{\|\boldsymbol{E}_i\|_F}-
\frac{\boldsymbol{E}_j}{\|\boldsymbol{E}_j\|_F}\|_F^2\bigg)\cdot\|\boldsymbol{E}_i\|_F\|\boldsymbol{E}_j\|_F \nonumber\\
\overset{(c)}{=}&\sum_{i=0}^1\sum_{j=2}^N\frac{1}{4}((1+\delta_{2k})^2
\left(\|\frac{\boldsymbol{E}_i}{\|\boldsymbol{E}_i\|_F}\|_F^2+\|\frac{\boldsymbol{E}_j}{\|\boldsymbol{E}_j\|_F}\|_F^2\right)
\nonumber \\
&-(1-\delta_{2k})^2\left(\|\frac{\boldsymbol{E}_i}{\|\boldsymbol{E}_i\|_F}\|_{F}^2+
\|\frac{\boldsymbol{E}_j}{\|\boldsymbol{E}_j\|_F}\|_F^2)\right)\cdot\|\boldsymbol{E}_i\|_F\|\boldsymbol{E}_j\|_F \nonumber\\
=&\sum_{i=0}^1\sum_{j=2}^N\frac{1}{2}((1+\delta_{2k})^2-
(1-\delta_{2k})^2)\cdot\|\boldsymbol{E}_i\|_F\|\boldsymbol{E}_j\|_F \nonumber\\
=&2\delta_{2k}\sum_{i=0}^1\sum_{j=2}^N\left\|\boldsymbol{E}_i\right\|_F\left\|\boldsymbol{E}_j\right\|_F\nonumber\\
=&2\delta_{2k}(\left\|\boldsymbol{E}_0\right\|_F+\left\|\boldsymbol{E}_1\right\|_F)
\sum_{j=2}^N\left\|\boldsymbol{E}_j\right\|_F\nonumber\\
\overset{(d)}{\leq}&2\delta_{2k}(\left\|\boldsymbol{E}_0\right\|_F+\left\|\boldsymbol{E}_1\right\|_F)
\|\boldsymbol{E}_0+\boldsymbol{E}_1\|_F \nonumber\\
\overset{(e)}{\leq}&2\sqrt{2}\delta_{2k}\left\|\boldsymbol{E}_0+\boldsymbol{E}_1\right\|_F^2
\label{lemma1:inequality}
\end{align}
where $(a)$ comes from the fact that
\begin{align}
\boldsymbol{A}\boldsymbol{E}\boldsymbol{B}^H =
\boldsymbol{A}\boldsymbol{X}\boldsymbol{B}^H
-\boldsymbol{A}\hat{\boldsymbol{X}}\boldsymbol{B}^H =
\boldsymbol{0}
\end{align}
$(b)$ follows from the equality
\begin{align}
4\Re\{\text{tr}(\boldsymbol{P}\boldsymbol{Q}^H)\} =
\|\boldsymbol{P+Q}\|_F^2-\|\boldsymbol{P-Q}\|_F^2
\end{align}
for any complex matrices $\boldsymbol{P}$ and $\boldsymbol{Q}$,
$(c)$ is due to the reason that $\boldsymbol{E}_i$ and
$\boldsymbol{E}_j$ have disjoint supports, $(d)$ follows from
(\ref{lemma1:eqn4}), and $(e)$ can be easily verified by noting
that
\begin{align}
\|\boldsymbol{E}_0+\boldsymbol{E}_1\|_F=
(\|\boldsymbol{E}_0\|_F^2+\|\boldsymbol{E}_1\|_F^2)^{1/2}
\end{align}

If $2\sqrt{2}\delta_{2k}-(1-\delta_{2k})^2<0$, i.e.
$\delta_{2k}<1+\sqrt{2}\left(1-\sqrt{1+\sqrt{2}}\right) $, then we
have $\left\|\boldsymbol{E}_0+\boldsymbol{E}_1\right\|_F=0$ from
(\ref{lemma1:inequality}), which implies that
$\left\|\boldsymbol{E}\right\|_F=0$, i.e.
$\boldsymbol{X}=\hat{\boldsymbol{X}}$. The proof is completed
here.

%\bibliography{newbib}

\begin{thebibliography}{10}
\providecommand{\url}[1]{#1}
\csname url@rmstyle\endcsname
\providecommand{\newblock}{\relax}
\providecommand{\bibinfo}[2]{#2}
\providecommand\BIBentrySTDinterwordspacing{\spaceskip=0pt\relax}
\providecommand\BIBentryALTinterwordstretchfactor{4}
\providecommand\BIBentryALTinterwordspacing{\spaceskip=\fontdimen2\font plus
\BIBentryALTinterwordstretchfactor\fontdimen3\font minus
  \fontdimen4\font\relax}
\providecommand\BIBforeignlanguage[2]{{%
\expandafter\ifx\csname l@#1\endcsname\relax
\typeout{** WARNING: IEEEtran.bst: No hyphenation pattern has been}%
\typeout{** loaded for the language `#1'. Using the pattern for}%
\typeout{** the default language instead.}%
\else
\language=\csname l@#1\endcsname
\fi
#2}}

\bibitem{RappaportMurdock11}
T.~S. Rappaport, J.~N. Murdock, and F.~Gutierrez, ``State of the art in
  60-{GH}z integrated circuits and systems for wireless communications,''
  \emph{Proc. IEEE}, vol.~99, no.~8, pp. 1390--1436, Aug. 2011.

\bibitem{RanganRappaport14}
S.~Rangan, T.~S. Rappaport, and E.~Erkip, ``Millimeter-wave cellular wireless
  networks: potentials and challenges,'' \emph{Proc. IEEE}, vol. 102, no.~3,
  pp. 366--385, March 2014.

\bibitem{GhoshThomas14}
A.~Ghosh, T.~A. Thomas, M.~C. Cudak, R.~Ratasuk, P.~Moorut, F.~W. Vook, T.~S.
  Rappaport, G.~R. MacCartney, S.~Sun, and S.~Nie, ``Millimeter-wave enhanced
  local area systems: a high-data-rate approach for future wireless networks,''
  \emph{IEEE J. Sel. Areas Commun.}, vol.~32, no.~6, pp. 1152--1163, June 2014.

\bibitem{SwindlehurstAyanoglu14}
A.~L. Swindlehurst, E.~Ayanoglu, P.~Heydari, and F.~Capolino, ``Millimeter-wave
  massive {MIMO}: the next wireless revolution?'' \emph{IEEE Commun. Mag.},
  vol.~52, no.~9, pp. 56--62, September 2014.

\bibitem{AlkhateebMo14}
A.~Alkhateeb, J.~Mo, N.~Gonzalez-Prelcic, and R.~Heath, ``{MIMO} precoding and
  combining solutions for millimeter-wave systems,'' \emph{IEEE Commun. Mag.},
  vol.~52, no.~12, pp. 122--131, December 2014.

\bibitem{HurKim13}
S.~Hur, T.~Kim, D.~J. Love, J.~V. Krogmeier, T.~A. Thomas, and A.~Ghosh,
  ``Millimeter wave beamforming for wireless backhaul and access in small cell
  networks,'' \emph{IEEE Trans. Commun.}, vol.~61, no.~10, pp. 4391--4403,
  October 2013.

\bibitem{AbariHassanieh16}
O.~Abari, H.~Hassanieh, M.~Rodriguez, and D.~Katabi, ``Millimeter wave
  communications: From point-to-point links to agile network connections,'' in
  \emph{Proc. 15th ACM Workshop on Hot Topics in Networks}, Atlanta, Georgia,
  USA, November 9-10 2016, pp. 169--175.

\bibitem{RamasamyVenkateswaran12a}
D.~Ramasamy, S.~Venkateswaran, and U.~Madhow, ``Compressive adaptation of large
  steerable arrays,'' in \emph{Proc. 2012 Information Theory and Applications
  Workshop (ITA)}, San Diego, California, USA, February 5-10 2012, pp.
  234--239.

\bibitem{RamasamyVenkateswaran12b}
------, ``Compressive tracking with 1000-element arrays: A framework for
  multi-gbps mm wave cellular downlinks,'' in \emph{Proc. 50th Annual Allerton
  Conference on Commun., Control, and Comput.}, October 2012, pp. 690--697.

\bibitem{AlkhateebyLeus15}
A.~Alkhateeb, G.~Leus, and R.~Heath, ``Compressed sensing based multi-user
  millimeter wave systems: How many measurements are needed?'' in \emph{Proc.
  40th IEEE Inter. Conf. on Acoust., Speech and Signal Process. (ICASSP)},
  Brisbane, Australia, April 19-24 2015, pp. 2909--2913.

\bibitem{AlkhateebAyach14}
A.~Alkhateeb, O.~E. Ayach, G.~Leus, and R.~Heath, ``Channel estimation and
  hybrid precoding for millimeter wave cellular systems,'' \emph{IEEE J. Sel.
  Topics Signal Process.}, vol.~8, no.~5, pp. 831--846, October 2014.

\bibitem{SchniterSayeed14}
P.~Schniter and A.~Sayeed, ``Channel estimation and precoder design for
  millimeter-wave communications: The sparse way,'' in \emph{Proc. 48th
  Asilomar Conf. Signals, Syst. Comput.}, Pacific Grove, California, USA,
  November 2-5 2014, pp. 273--277.

\bibitem{KimLove15}
T.~Kim and D.~J. Love, ``Virtual {AoA} and {AoD} estimation for sparse
  millimeter wave {MIMO} channels,'' in \emph{Proc. 16th IEEE Inter. Workshop
  on Signal Process. Advances in Wireless Commun. (SPAWC)}, Stockholm, Sweden,
  June 28 - July 1 2015, pp. 146--150.

\bibitem{MarziRamasamy16}
Z.~Marzi, D.~Ramasamy, and U.~Madhow, ``Compressive channel estimation and
  tracking for large arrays in mm-{W}ave picocells,'' \emph{IEEE J. Sel. Topics
  Signal Process.}, vol.~10, no.~3, pp. 514--527, April 2016.

\bibitem{GaoDai15b}
Z.~Gao, L.~Dai, Z.~Wang, and S.~Chen, ``Spatially common sparsity based
  adaptive channel estimation and feedback for {FDD} massive {MIMO},''
  \emph{IEEE Trans. Signal Process.}, vol.~63, no.~23, pp. 6169--6183, December
  2015.

\bibitem{GaoDai16}
X.~Gao, L.~Dai, and A.~M. Sayeed, ``Low {RF}-complexity technologies for 5{G}
  millimeter-wave {MIMO} systems with large antenna arrays,'' \emph{available
  at arXiv:1607.04559}, 2016.

\bibitem{ZhouFang16}
Z.~Zhou, J.~Fang, L.~Yang, H.~Li, Z.~Chen, and S.~Li, ``Channel estimation for
  millimeter-wave multiuser {MIMO} systems via {PARAFAC} decomposition,''
  \emph{IEEE Trans. Wireless Commun.}, vol.~15, no.~11, pp. 7501--7516,
  November 2016.

\bibitem{ZhouFang17}
Z.~Zhou, J.~Fang, L.~Yang, H.~Li, Z.~Chen, and R.~S. Blum, ``Low-rank tensor
  decomposition-aided channel estimation for millimeter wave mimo-ofdm
  systems,'' \emph{IEEE Journal Selected Areas in Communications}, to appear.

\bibitem{SamimiWang13}
M.~Samimi, K.~Wang, Y.~Azar, G.~N. Wong, R.~Mayzus, H.~Zhao, J.~K. Schulz,
  S.~Sun, F.~Gutierrez, and T.~S. Rappaport, ``28 {GHz} angle of arrival and
  angle of departure analysis for outdoor cellular communications using
  steerable beam antennas in {New} {York} city,'' in \emph{Proc. 2013 IEEE 77th
  Vehicular Technology Conference (VTC Spring)}, Dresden, Germany, June 2-5
  2013, pp. 1--6.

\bibitem{ZhaoMayzus13}
H.~Zhao, R.~Mayzus, S.~Sun, M.~Samimi, J.~K. Schulz, Y.~Azar, K.~Wang, G.~N.
  Wong, F.~Gutierrez, and T.~S. Rappaport, ``28 {GHz} millimeter wave cellular
  communication measurements for reflection and penetration loss in and around
  buildings in {New} {York} city,'' in \emph{Proc. 2013 IEEE International
  Conference on Communications (ICC)}, Budapest, Hungary, June 9-13 2013, pp.
  5163--5167.

\bibitem{RappaportSun13}
T.~S. Rappaport, S.~Sun, R.~Mayzus, H.~Zhao, Y.~Azar, K.~Wang, G.~N. Wong,
  J.~K. Schulz, M.~Samimi, and F.~Gutierrez, ``Millimeter wave mobile
  communications for 5{G} cellular: It will work!'' \emph{IEEE Access}, vol.~1,
  pp. 335--349, May 2013.

\bibitem{AkdenizLiu14}
M.~R. Akdeniz, Y.~Liu, M.~K. Samimi, S.~Sun, S.~Rangan, T.~S. Rappaport, and
  E.~Erkip, ``Millimeter wave channel modeling and cellular capacity
  evaluation,'' \emph{IEEE J. Sel. Areas Commun.}, vol.~32, no.~6, pp.
  1164--1179, June 2014.

\bibitem{WangPajovic17}
P.~Wang, M.~Pajovic, P.~V. Orlik, T.~Koike-Akino, K.~J. Kim, and J.~Fang,
  ``Sparse channel estimation in millimeter wave communications: Exploiting
  joint {AoD}-{AoA} angular spread,'' in \emph{Proc. 2017 IEEE International
  Conference on Communications (ICC)}, Paris, France, May 21-25 2017.

\bibitem{BeckTeboulle09}
A.~Beck and M.~Teboulle, ``A fast iterative shrinkage-thresholding algorithm
  for linear inverse problems,'' \emph{SIAM J. Imaging Sci.}, vol.~2, no.~1,
  pp. 183--202, March 2009.

\bibitem{AlkhateebHeath16}
A.~Alkhateeb and R.~W. Heath, ``Frequency selective hybrid precoding for
  limited feedback millimeter wave systems,'' \emph{IEEE Trans. Commun.},
  vol.~64, no.~5, pp. 1801--1818, May 2016.

\bibitem{CandesTao05}
E.~Cand\'{e}s and T.~Tao, ``Decoding by linear programming,'' \emph{IEEE Trans.
  Information Theory}, no.~12, pp. 4203--4215, Dec. 2005.

\bibitem{TroppGilbert07}
J.~A. Tropp and A.~C. Gilbert, ``Signal recovery from random measurements via
  orthogonal matching pursuit,'' \emph{IEEE Trans. Information Theory},
  vol.~53, no.~12, pp. 4655--4666, Dec. 2007.

\bibitem{CandesRecht09}
E.~J. Cand\`{e}s and B.~Recht, ``Exact matrix completion via convex
  optimization,'' \emph{Foundations of Computational Mathematics}, vol.~9,
  no.~6, pp. 717--772, December 2009.

\bibitem{RechtFazel10}
B.~Recht, M.~Fazel, and P.~A. Parrilo, ``Guaranteed minimum-rank solutions of
  linear matrix equations via nuclear norm minimization,'' \emph{SIAM Review},
  vol.~52, no.~3, pp. 471--501, August 2010.

\bibitem{KoltchinskiiLounici11}
V.~Koltchinskii, K.~Lounici, and A.~B. Tsybakov, ``Nuclear-norm penalization
  and optimal rates for noisy low-rank matrix completion,'' \emph{The Annals of
  Statistics}, vol.~39, no.~5, pp. 2302--2329, October 2011.

\bibitem{BahmaniRomberg15}
S.~Bahmani and J.~Romberg, ``Efficient compressive phase retrieval with
  constrained sensing vectors,'' in \emph{Advances in Neural Information
  Processing Systems (NIPS)}, vol.~28, Montreal, Quebec, Canada, December 7-12
  2015, pp. 523--531.

\bibitem{BaraniukDavenport08}
R.~Baraniuk, M.~Davenport, R.~DeVore, and M.~Wakin, ``A simple proof of the
  restricted isometry property for random matrices,'' \emph{Constructive
  Approximation}, vol.~28, no.~3, pp. 253--263, January 2008.

\bibitem{DuarteEldar11}
M.~F. Duarte and Y.~C. Eldar, ``Structured compressed sensing: From theory to
  applications,'' \emph{IEEE Transactions on Signal Processing}, vol.~59,
  no.~9, pp. 4053--4085, September 2011.

\bibitem{MaGoldfarb11}
S.~Ma, D.~Goldfarb, and L.~Chen, ``Fixed point and bregman iterative methods
  for matrix rank minimization,'' \emph{Mathematical Programming}, vol. 128,
  no.~1, pp. 321--353, June 2011.

\bibitem{CaiCandes10}
J.~Cai, E.~J. Cand\`{e}s, and Z.~Shen, ``A singular value thresholding
  algorithm for matrix completion,'' \emph{SIAM Journal on Optimization},
  vol.~20, no.~4, pp. 1956--1982, March 2010.

\end{thebibliography}
%\bibliographystyle{IEEEtran}

\end{document}